\pdfoutput=1
\RequirePackage{ifpdf}
\ifpdf % We are running pdfTeX in pdf mode
\documentclass[pdftex]{sigma}
\else
\documentclass{sigma}
\fi

\numberwithin{equation}{section}

\newtheorem{Theorem}{Theorem}[section]
\newtheorem{Corollary}[Theorem]{Corollary}
\newtheorem{Lemma}[Theorem]{Lemma}
\newtheorem{Proposition}[Theorem]{Proposition}
 { \theoremstyle{definition}
\newtheorem{Definition}[Theorem]{Definition}
\newtheorem{Remark}[Theorem]{Remark} }

\newcommand{\abs}[1]{\left\lvert #1 \right\rvert }

\newmuskip\pFqskip
\mathchardef\pFcomma=\mathcode`, % keep a copy of the comma

\newcommand*\pFq[5]{%
 \begingroup
 \begingroup\lccode`~=`,
 \lowercase{\endgroup\def~}{\pFcomma\mkern\pFqskip}%
 \mathcode`,=\string"8000
 {}_{#1}F_{#2}\biggl[\genfrac..{0pt}{}{#3}{#4};#5\biggr]%
 \endgroup
}

\newcommand{\dbyd}[2][]{\frac{d #1}{d #2}}
\newcommand{\pdbyd}[2][]{\frac{\partial #1}{\partial #2}}

\renewcommand{\geq}{\geqslant}
\renewcommand{\leq}{\leqslant}
\newcommand{\bn}{\boldsymbol{n}}

\begin{document}

\allowdisplaybreaks

\newcommand{\arXivNumber}{1508.06689}

\renewcommand{\thefootnote}{$\star$}

\renewcommand{\PaperNumber}{079}

\FirstPageHeading

\ShortArticleName{Hypergeometric Integral for Hypersphere Fundamental Solutions}

\ArticleName{A Hypergeometric Integral with Applications \\ to the Fundamental Solution of Laplace's Equation \\ on Hyperspheres\footnote{This paper is a~contribution to the Special Issue on Orthogonal Polynomials, Special Functions and Applications. The full collection is available at \href{http://www.emis.de/journals/SIGMA/OPSFA2015.html}{http://www.emis.de/journals/SIGMA/OPSFA2015.html}}}

\Author{Richard CHAPLING}
\AuthorNameForHeading{R.~Chapling}

\Address{Department of Applied Mathematics and Theoretical Physics,\\ University of Cambridge, Cambridge, England}
\Email{\href{mailto:rc476@cam.ac.uk}{rc476@cam.ac.uk}}

\ArticleDates{Received November 23, 2015, in f\/inal form August 04, 2016; Published online August 10, 2016}

\Abstract{We consider Poisson's equation on the $n$-dimensional sphere in the situation where the inhomogeneous term has zero integral. Using a number of classical and modern hypergeometric identities, we integrate this equation to produce the form of the fundamental solutions for any number of dimensions in terms of generalised hypergeometric functions, with dif\/ferent closed forms for even and odd-dimensional cases.}


\Classification{35A08; 35J05; 31C12; 33C05; 33C20}

\renewcommand{\thefootnote}{\arabic{footnote}}
\setcounter{footnote}{0}

\section{Introduction}
The calculation of fundamental solutions is dif\/f\/icult even for simple geometries. However, in the case of Laplace's equation on the $n$-dimensional sphere, the sphere's homogeneity allows us to reduce the partial dif\/ferential equation to an ordinary one, rendering the computation tractable. The author is aware of two previous calculations of this type in the literature, \cite{Cohl:2011vn} and \cite{Szmytkowski:2007gf}, which disagree. In this paper, we shall proceed along similar lines to \cite{Cohl:2011vn}; in the process, it will become apparent that the calculation in \cite{Cohl:2011vn} is in fact the correct solution to a subtly dif\/ferent problem. We shall also check the consistency of our results with those of \cite{Szmytkowski:2007gf} in the case $n=2$.

The even- and odd-dimensional cases are treated separately: one may be expressed as a~hyper\-geometric polynomial in $\cot{\tfrac{1}{2} \theta}$, the other as a sum of an even polynomial in $\cot{\theta}$ and an odd polynomial multiplied by $\pi-\theta$. There is reason to expect this disparity, given the other ways in which spheres of even and odd dimension are dif\/ferent, such as the hairy ball theorem and the volume in terms of factorials.

The last section gives some applications of our result, including the dipole potential on the sphere, the azimuthal Fourier expansion, fundamental solutions on real projective space by considering it as a quotient of the sphere, and a similar calculation explaining the problem that~\cite{Cohl:2011vn} is actually solving.

\section{The generalised Green's function}\label{sec:ggf}

Let $M$ be a compact manifold, $ S^{n} $ the (unit) $n$-dimensional hypersphere, and $ S^{n}_{R} $ the hypersphere of radius $R$.

In a recent paper \cite{Cohl:2011vn} is what is claimed to be a calculation of a fundamental solution of Laplace's equation on hyperspheres. The starting point of that paper is the equation
\begin{gather}\label{eq:notgf}
	-\Delta{G(x,x')} = \delta(x,x'),
\end{gather}
where $x,x' \in S^{n}_{R}$, and $-\Delta $ is the Laplace--Beltrami operator. However, this is not a well-posed equation on the sphere (or, indeed, any compact manifold)\footnote{On the other hand, the later paper~\cite{Cohl:2012ly} is unaf\/fected, since the hyperbolic manifolds therein are not compact.}; the best way to demonstrate this is to consider the Poisson equation that we normally use fundamental solutions to solve:
\begin{gather}\label{eq:Poisson}
	-\Delta{u} = \rho,
\end{gather}
where $\rho \in L^{2}(S^{n}_{R}) \cap C^{1}(S^{n}_{R})$.

The divergence theorem,
\begin{gather*} \int_{M} ( \operatorname{div}{X} ) \, dV_{g} = \int_{\partial M} \langle X , N \rangle_{g} \, dV_{\tilde{g}} \end{gather*}
for any $ u \in C^{\infty}(M) $ \cite[Theorem~16.32]{Lee:2013ys}\footnote{This is, of course, just a recasting of the fundamental theorem of exterior calculus, $\int_{M} d\omega = \int_{\partial M} \omega$.} applied to $\operatorname{grad}{u}$ shows that
\begin{gather*} \int_{S^{n}_{R}} (-\Delta){u} = 0, \end{gather*}
because $\partial S^{n}_{R} = \varnothing$; (\ref{eq:Poisson}) then implies that
\begin{gather*} \int_{S^{n}_{R}} \rho = 0. \end{gather*}

If we interpret \eqref{eq:notgf} distributionally, we should have
\begin{align*}
	\langle -\Delta_{x} \langle G(x;y),\varphi(y) \rangle_{y}, \chi(x) \rangle_{x} &= \langle \langle G(x;y),\varphi(y) \rangle_{y}, -\Delta_{x}\chi(x) \rangle_{x} \\
	&= \langle \langle \delta_{x}(y),\varphi(y) \rangle_{y},\chi(x) \rangle_{x}  = \langle \varphi(x),\chi(x) \rangle_{x} = \int_{M} \varphi \chi,
\end{align*}
for every pair of test functions $ \varphi,\chi \in \mathcal{D} := C^{\infty}(M) $, where the f\/irst equality is the def\/inition of the Laplacian as a distributional derivative, the second the def\/inition of $G$ as a distribution, the third the def\/inition of~$ \delta_{x}$, and the fourth the def\/inition of the pairing for smooth functions.

If we take $ \chi(x) \equiv 1 $ in the above equation, we arrive at a contradiction, since the f\/irst line must be zero and the last line need not be; hence, we need to modify the starting point, and consider a dif\/ferent initial def\/inition of $ G$. Suppose that we instead include a constant on the right-hand side of \eqref{eq:notgf}:
\begin{gather*}%\label{eq:ggf}
	-\Delta{G(x,x')} = \delta(x,x')-a.
\end{gather*}
Carrying out the same procedure, we obtain
\begin{gather*}
	\langle -\Delta_{x} \langle G(x;y),\varphi(y) \rangle_{y}, \chi(x) \rangle_{x}  = \langle \langle \delta_{x}(y)-a,\varphi(y) \rangle_{y},\chi(x) \rangle_{x} = \int_{M} \varphi \chi-a\left( \int_{M}\! \varphi \right) \left( \int_{M} \!\chi \right),
\end{gather*}
which can only be consistent with the $\chi(x) \equiv 1$ case if $a=1/\abs{M}$, where $\abs{M}$ is the volume of~$M$, i.e.,
\begin{gather*} \abs{M} = \int_{M}dV_{g}. \end{gather*}
The behaviour is then that required of the usual Green's function in the space of functions orthogonal to~$1$.

See also \cite[p.~354~f\/f.]{Courant:1953zr}, wherein this is discussed as the notion of ``Green's function in the generalised sense''. Much of the theory carries straight over, provided that $ \delta $ is replaced by $ \delta-\abs{M}^{-1} $ in all cases\footnote{See, for example, \cite[p.~108~f\/f.]{Aubin:1998vn}.}.

Lastly in our general discussion, we note that a function $G$ with these properties need not be unique: in particular, if we have two solutions, $ G$, $G'$, then
\begin{gather*}
	-\Delta(G-G')=\big(\delta-\abs{M}^{-1}\big)-\big(\delta-\abs{M}^{-1}\big)=0,
\end{gather*}
so the dif\/ference is a harmonic function. The maximum principle implies that the only harmonic functions on a compact manifold are constant\footnote{See, for example,~\cite{yamasuge1957}.}, so $G$ is determined up to an additive constant, which we may choose freely. Our convention shall be to choose the constant so that~$G$ is nonnegative on~$S^{n}_{R}$; we shall construct it to be so with a zero at the antipodal point to~$x$.

\section[The ordinary dif\/ferential equation for $G_{n}$]{The ordinary dif\/ferential equation for $\boldsymbol{G_{n}}$}

\subsection{The Laplace--Beltrami operator on a hyperspherical manifold}
To be explicit, let the Green's function for the $ n $-dimensional hypersphere be denoted~$G_{n}$. One standard coordinate chart on the $n$-dimensional hypersphere is $(\theta_{i},\varphi)$, $0 < \theta_{i} < \pi$, $0 < \varphi < 2\pi$; this shall suf\/f\/ice for our purposes. (Our derivation does not depend on $\varphi$, and excluding the point antipodal to $ \theta=0 $ will also not matter, since the function is constructed to be smooth there.) The metric is given by \cite[p.~52, Example~7.4]{iliev2006handbook}
\begin{gather*}
	g = \operatorname{diag}\big(R^{2}\sin^{2}{\theta_{1}},R^{2}\sin^{2}{\theta_{1}}\sin^{2}{\theta_{2}}, \dotsc\big).
\end{gather*}

Taking $\theta_{1}=\theta$, since a fundamental solution in a homogeneous set does not depend on direction, our fundamental solution reduces to $G_{n}=G_{n}(\theta)$, so we f\/ind that the Laplacian reduces to
\begin{gather*} \Delta{f(\theta)} = \frac{1}{R^{2}\sin^{n-1}{\theta}} \dbyd{\theta} \big( \sin^{n-1}{\theta} f'(\theta) \big). \end{gather*}

Therefore we need to solve the ordinary dif\/ferential equation
\begin{gather}\label{eq:de}
	f'' + (n-1) \cot{\theta} f' = aR^{2},
\end{gather}
for $\theta \in (0,\pi)$, where~$a$ is the constant def\/ined in the previous section. (Since the interval does not contain $ \theta=0 $, the $ \delta(\theta) $ term is not present.) The theory in the previous section dictates that, with~$ S_{n} $ the volume of the unit $ n $-sphere,
\begin{gather}\label{eq:constant}
	a = \frac{1}{S_{n}R^{n}} = \frac{\Gamma((n+1)/2)}{2\pi^{(n+1)/2}R^{n}},
\end{gather}	
which is the reciprocal of the volume of the $n$-sphere; we discuss the relationship of this global constraint with the behaviour as $\theta \downarrow 0$ in Section~\ref{sec:ccc} below.

\subsection{Formal solution}
Recall that a f\/irst-order dif\/ferential equation
\begin{gather*} y' + g y = u \end{gather*}
may be written in the form
\begin{gather*} e^{-g} (ye^{g})' = u. \end{gather*}
It is then simple to show that for suitable limits on the integrals,
\begin{gather*} y(x) = e^{-g(x)} \int_{x_{0}}^{x} e^{g(t)} u(t) \, dt \end{gather*}
is the solution, where $x_{0}$ is some constant.

We conclude that (\ref{eq:de}), with $a$ as in \eqref{eq:constant}, has the formal solution
\begin{gather}\label{eq:formalsoln}
	G_{n}(\theta) = \frac{1}{S_{n}R^{n-2}} \int_{\theta}^{\pi} \csc^{n-1}{\phi} \int_{\phi}^{\pi} \sin^{n-1}{\psi} \, d\psi \, d\phi.
\end{gather}
(Since the integrand is a smooth function of $ \phi $ for $0<\phi \leq \pi$, the integral is smooth and bounded on $[\varepsilon,\pi]$ for any suf\/f\/iciently small $\varepsilon>0$.)

\subsection[The behaviour as $\theta \downarrow 0$: consistency with fundamental solutions on f\/lat space]{The behaviour as $\boldsymbol{\theta \downarrow 0}$: consistency with fundamental solutions\\ on f\/lat space}\label{sec:ccc}

In this section we consider the limit for $n \neq 2$; the two-dimensional case is analogous using $ \log{s} $. The integral~\eqref{eq:formalsoln} was determined by the global behaviour of fundamental solutions (\emph{viz.}\ smoothness away from the pole at $\theta=0$). On the other hand, the arc-length on the hypersphere of radius $R$ is given by $ s=R\theta $, so we expect to f\/ind that in the f\/lat-space limit as $R \to \infty$,
\begin{gather*}
	G_{n}(s/R) \sim \frac{1}{(n-2)S_{n-1}s^{n-2}},
\end{gather*}
which shall give us the local behaviour as $\theta \downarrow 0$. We wish to verify that our solution has this local behaviour. To agree with the f\/lat-space fundamental solution, our fundamental solution should also exhibit a singularity of order $\theta^{2-n}$ near $\theta=0$, with coef\/f\/icient $1/((n-2)S_{n-1}R^{n-2})$. We therefore need to have
\begin{gather*}
	\frac{1}{S_{n}R^{n-2}}\int_{\theta}^{\pi} \csc^{n-1}{\phi} \int_{\phi}^{\pi} \sin^{n-1}{\psi} \, d\psi \, d\phi \sim \frac{1}{(n-2)S_{n-1}(R\theta)^{n-2}}
\end{gather*}
as $\theta \downarrow 0$. We can drop the part of the $\phi$ integral from $\pi/2$ to $\pi$ as irrelevant, it being just a~constant, and then we need to examine
\begin{gather*}
	\frac{1}{S_{n}R^{n-2}}\int_{\theta}^{\pi/2} \csc^{n-1}{\phi} \int_{\phi}^{\pi} \sin^{n-1}{\psi} \, d\psi \, d\phi.
\end{gather*}
It is apparent that as $\theta \downarrow 0$, the integrand behaves as
\begin{gather*}
	\csc^{n-1}{\phi} \int_{\phi}^{\pi} \sin^{n-1}{\psi} \, d\psi \sim \phi^{1-n} \int_{0}^{\pi} \sin^{n-1}{\psi} \, d\psi = \phi^{1-n} \frac{\pi^{1/2}\Gamma(n/2)}{\Gamma((n+1)/2)}.
\end{gather*}
A result in asymptotic analysis shows that if $F$ is suf\/f\/iciently well-behaved and $ F(x) \sim f(x) $ as $x \downarrow x_{0}$, then for $ x \in (x_{0},x_{0}+\varepsilon)$, $\varepsilon>0$, we have $\int_{x}^{x_{0}+\varepsilon} F \sim \int_{x}^{x_{0}+\varepsilon} f $ as $x \downarrow x_{0}$. It is easy to see that the integrand behaves as such, and hence we f\/ind
\begin{align*}
	\frac{R^{2-n}}{S_{n}}\int_{\theta}^{\pi/2} \csc^{n-1}{\phi} \int_{\phi}^{\pi} \sin^{n-1}{\psi} \, d\psi \, d\phi \sim \frac{\theta^{2-n}R^{2-n}}{(n-2)S_{n}} \frac{\pi^{1/2}\Gamma(n/2)}{\Gamma((n+1)/2)} = \frac{(R\theta)^{2-n}}{(n-2)S_{n-1}},
\end{align*}
as expected, using the recurrence relation for the hyperspherical volumes.

\section{Calculating explicit expressions for our fundamental solution}
In the following sections we shall consider two separate ways of determining our fundamental solution explicitly: the f\/irst proves that the solutions are f\/inite sums of various trigonometric functions by deriving recurrence relations, the second uses hypergeometric functions and identities to produce closed form expressions for these sums, and recurrence relations that determine the explicit form of the solution. Even and odd dimensions are treated separately since the form of the solution expression is quite dif\/ferent in each case.

We shall use the notation
\begin{gather}\label{eq:ImJmdefs}
	I_{m}(\phi) := \int_{\phi}^{\pi} \sin^{m-1}{\psi} \, d\psi, \qquad
	J_{m}(\theta) := \int_{\theta}^{\pi} \frac{I_{m}(\phi)}{\sin^{m-1}{\phi}} \, d\phi.
\end{gather}
 It is apparent that $J_{m}(\theta)$ shall be proportional to our fundamental solution: in particular,
 \begin{gather*}
 	G_{n}(\theta) = \frac{R^{2-n}}{S_{n}} J_{n}(\theta).
 \end{gather*}

\subsection{Recurrence relation}\label{sec:rr}

\begin{Proposition}
The integral $ I_{m}(\phi) $ satisfies the recurrence relation
\begin{gather}	\label{eq:Irec}
I_{m}(\phi) = \frac{m-2}{m-1} I_{m-2}(\phi) + \frac{1}{m-1} \sin^{m-2}{\phi} \cos{\phi}
\end{gather}
for $m>0$, with basis cases
\begin{gather*}
I_{1}(\phi) = \pi-\phi, \qquad I_{2}(\phi) = 1 + \cos{\phi}.
\end{gather*}
\end{Proposition}

\begin{proof}It is easy to verify the basis cases directly. We can carry out a standard calculation to f\/ind
\begin{align*}
I_{m}(\phi) &= \int_{\phi}^{\pi} \sin^{m-3}{\psi}\big(1-\cos^{2}{\psi}\big) \, d\psi
= I_{m-2}(\phi) - \int_{\phi}^{\pi} \sin^{m-3}{\psi}\cos^{2}{\psi} \, d\psi \\
&= I_{m-2}(\phi) - \left[\frac{1}{m-2} \sin^{m-2}{\psi} \cos{\psi}\right]_{\phi}^{\pi} + \frac{1}{m-2}\int_{\phi}^{\pi} \sin^{m-2}{\psi}(-\sin{\psi}) \, d\psi \\
&= I_{m-2}(\phi) + \frac{1}{m-2} \sin^{m-2}{\phi} \cos{\phi} - \frac{1}{m-2} I_{m}(\phi),
\end{align*}
and we therefore have
\begin{gather*}
	I_{m}(\phi) = \frac{m-2}{m-1} I_{m-2}(\phi) + \frac{1}{m-1} \sin^{m-2}{\phi} \cos{\phi}
\end{gather*}
for $m>0$, as required.
\end{proof}

We may then verify, for example,
\begin{gather*} I_{3}(\phi) = \frac{2-1}{2} I_{1}(\phi) + \frac{1}{2} \sin^{2-1}{\phi}\cos{\phi} = \frac{1}{2}(\pi-\phi) + \frac{1}{2}\sin{\phi}\cos{\phi}. \end{gather*}
We can use this to f\/ind a recurrence relation for $J_{m}(\theta)$.

\begin{Proposition}
For $m>0$, the integral $ J_{m} $ satisfies the recurrence relation
\begin{gather*}
J_{m}(\theta) = \frac{m-3}{m-1} J_{m-2}(\theta) + \frac{1}{m-1} \frac{\cos{\theta}}{\sin^{m-2}{\theta}} I_{m-2}(\theta) + \frac{1}{(m-1)(m-2)}
\end{gather*}
with basis cases
\begin{gather*}
J_{1}(\theta) = \frac{1}{2}(\pi - \theta)^{2}, \qquad J_{2}(\theta) = \log{\csc^{2}{\tfrac{1}{2}\theta}}.
\end{gather*}
\end{Proposition}

\begin{proof}
The basis cases are simple: for example,
\begin{gather*}
	J_{2}(\theta) = \int_{\theta}^{\pi} \frac{1+\cos{\phi}}{\sin{\phi}} \, d\phi = \int_{\theta}^{\pi} \cot{\tfrac{1}{2}\phi} \, d\phi = \log{\csc^{2}{\tfrac{1}{2}\theta}}.
\end{gather*}
For the general case, \cite{Cohl:2011vn} notes the useful trigonometric/partial fraction identity
\begin{gather*} \frac{1}{\sin^{m}{\phi}} = \frac{1}{\sin^{m-2}{\phi}} + \frac{\cos^{2}{\phi}}{\sin^{m}{\phi}}, \end{gather*}
which is a trivial consequence of Pythagoras's identity. We therefore write
\begin{gather*}
	J_{m}(\theta) = \int_{\theta}^{\pi} \frac{I_{m}(\phi)}{\sin^{m-3}{\phi}} \, d\phi + \int_{\theta}^{\pi} \frac{I_{m}(\phi)\cos^{2}{\phi}}{\sin^{m-1}{\phi}} \, d\phi.
\end{gather*}
The next step is the most delicate part: we integrate the last integral by parts, being careful to note that near $\theta=\pi$, $I_{m}(\phi) \sim \frac{1}{m}(\pi-\phi)^{m}$:
\begin{gather*}
	\int_{\theta}^{\pi} \frac{\cos{\phi}}{\sin^{m-1}{\phi}}I_{m}(\phi)\cos{\phi} \, d\phi \\
\qquad{} = \left[-\frac{1}{m-2}\frac{I_{m}(\phi)\cos{\phi}}{\sin^{m-2}{\phi}}\right]_{\theta}^{\pi} + \frac{1}{m-2} \int_{\theta}^{\pi}\frac{I'_{m}(\phi)\cos{\phi}-I_{m}(\phi)\sin{\phi}}{\sin^{m-2}{\phi}} \, d\phi \\
\qquad{} = \frac{1}{m-2}\frac{\cos{\theta}}{\sin^{m-2}{\theta}} I_{m}(\phi) - \frac{1}{m-2} \int_{\theta}^{\pi}\frac{\sin^{m-1}{\phi}\cos{\phi}}{\sin^{m-2}{\phi}} \, d\phi - \frac{1}{m-2} \int_{\theta}^{\pi}\frac{I_{m}(\phi)}{\sin^{m-3}{\phi}} \, d\phi \\
\qquad{} = \frac{1}{m-2}\frac{\cos{\theta}}{\sin^{m-2}{\theta}} I_{m}(\phi) - \frac{1}{m-2} \int_{\theta}^{\pi}\sin{\phi}\cos{\phi} \, d\phi - \frac{1}{m-2} \int_{\theta}^{\pi}\frac{I_{m}(\phi)}{\sin^{m-3}{\phi}} \, d\phi.
\end{gather*}
We therefore have
\begin{gather*}
	J_{m}(\theta) = \frac{1}{m-2} \left((m-3) \int_{\theta}^{\pi} \frac{I_{m}(\phi)}{\sin^{m-3}{\phi}} \, d\phi + \frac{\cos{\theta}}{\sin^{m-2}{\theta}} I_{m}(\theta) - \int_{\theta}^{\pi}\sin{\phi}\cos{\phi} \, d\phi \right).
\end{gather*}

The f\/irst integral is substantially easier: it splits as
\begin{align*}
	\int_{\theta}^{\pi} \frac{I_{m}(\phi)}{\sin^{m-3}{\phi}} \, d\phi &= \frac{m-2}{m-1}\int_{\theta}^{\pi} \frac{I_{m-2}(\phi)}{\sin^{m-3}{\phi}} \, d\phi + \frac{1}{m-1}\int_{\theta}^{\pi} \frac{\sin^{m-2}{\phi} \cos{\phi}}{\sin^{m-3}{\phi}} \, d\phi \\
&= \frac{m-2}{m-1} J_{m-2}(\theta) + \frac{1}{m-1}\int_{\theta}^{\pi} \sin{\phi} \cos{\phi} \, d\phi,
\end{align*}
and hence
\begin{gather*}
	J_{m}(\theta) = \frac{m-3}{m-1} J_{m-2}(\theta) + \frac{1}{m-2} \frac{\cos{\theta}}{\sin^{m-2}{\theta}} I_{m}(\theta) -\frac{2}{(m-1)(m-2)}\int_{\theta}^{\pi} \sin{\phi} \cos{\phi} \, d\phi.
\end{gather*}
Computing the last integral, we obtain the recurrence relation
\begin{gather*}
	J_{m}(\theta) = \frac{m-3}{m-1} J_{m-2}(\theta) + \frac{1}{m-2} \frac{\cos{\theta}}{\sin^{m-2}{\theta}} I_{m}(\theta) + \frac{1}{(m-1)(m-2)}\sin^{2}{\theta}.
\end{gather*}

We can simplify the answer by using (\ref{eq:Irec}) to give
\begin{gather*}
	J_{m}(\theta) = \frac{m-3}{m-1} J_{m-2}(\theta) + \frac{1}{m-1} \frac{\cos{\theta}}{\sin^{m-2}{\theta}} I_{m-2}(\theta) + \frac{1}{(m-1)(m-2)},
\end{gather*}
as required.
\end{proof}

We shall now consider another way to obtain this recurrence, which will additionally produce closed-form results.

\subsection{Hypergeometric functions}

\subsubsection[$n$ even: stereographic projection]{$\boldsymbol{n}$ even: stereographic projection}

Note that stereographic projection $S^{n} \to {\mathbb R}^{n}$ can be def\/ined in terms of spherical coordinates $(\theta, \theta_{2}, \dotsc, \theta_{n-1} , \varphi)$ on $S^{n}$ and spherical coordinates $(r, \theta'_{2}, \dotsc, \theta'_{n-1} , \varphi')$ on ${\mathbb R}^{n}$ by
\begin{gather*}
	r = \cot{\tfrac{1}{2} \theta }, \qquad \theta'_{i} = \theta_{i}, \quad 2\leq i \leq n-1, \qquad \varphi'= \varphi.
\end{gather*}

The metric is transformed as
\begin{gather*}
	d\theta^{2} + \sin^{2}{\theta}(d\Sigma_{n-1})^{2} = \frac{4}{(1+r^{2})^{2}} \left( dr^{2} + r^{2} (d\Sigma'_{n-1})^{2} \right),
\end{gather*}
where $ (d\Sigma_{n-1})^{2} $ is the metric on the $(n-1)$-dimensional sphere def\/ined by $(\theta_{2}, \theta_{3}, \dotsc, \theta_{n-1} , \varphi) $, and likewise $ (d\Sigma_{n-1}')^{2} $ in the primed variables. The Laplacian becomes
\begin{align*}
	\Delta{f} &= \frac{1}{\sin^{n-1}{\theta}} \pdbyd{\theta} \left(\sin^{n-1}{\theta} \pdbyd[f]{\theta} \right) + \frac{1}{\sin^{2}{\theta}} \Delta_{\Sigma_{n-1}}{f} \\
	&= \frac{(1+r^{2})^{n}}{4r^{n-1}} \pdbyd{r} \left( \frac{r^{n-1}}{(1+r^{2})^{n-2}} \pdbyd[f]{r} \right) + \frac{1}{r^{2}} \Delta_{\Sigma'_{n-1}}{f},
\end{align*}
where $ \Delta_{\Sigma_{n-1}} $ is the Laplacian restricted to the $(n-1)$-dimensional sphere, and likewise $ \Delta_{\Sigma'_{n-1}} $ in the primed variables.

Therefore in the new coordinates, the Green's function is given by the integral expression
\begin{gather}\label{eq:stereogf}
	G_{n}(r) = 4\frac{R^{2-n}}{S_{n}} \int_{0}^{r} \frac{(1+x^{2})^{n-2}}{x^{n-2}} \left( \int_{0}^{x} \frac{y^{n}}{(1+y^{2})^{n}} \, \frac{dy}{y} \right) \frac{dx}{x}.
\end{gather}

We observe that projection from the antipodal point is just $r' = \tan{\tfrac{1}{2} \theta = 1/r}$. This fact is neatly captured in the $r \mapsto 1/r$ symmetry of the integrands in \eqref{eq:stereogf}. Using this expression, we shall prove

\begin{Theorem}	The Green's function for the $n$-dimensional sphere of radius $R$ is
\begin{gather*}%\label{eq:gfformula}
G_{n}(\theta) = \frac{R^{2-n}}{S_{n}} \left( \frac{n-2}{n(n-1)} r^{2} \, \pFq{3}{2}{1,1,2-n/2}{2,1+n/2}{-r^{2}} + \frac{1}{n-1} \log\big(1+r^{2}\big) \right),
\end{gather*}
where $r=\cot{\frac{1}{2}\theta}$.
\end{Theorem}

\begin{proof}
Consider f\/irst the inner integral in \eqref{eq:stereogf}. This can easily be transformed to one of hypergeometric Euler type\footnote{This integral was not actually discovered by Euler: the f\/irst unmistakeable linking of this integral to the hypergeometric series appears in the doctoral dissertation of P.O.C.~Vorsselman de Heer~\cite[p.~10]{VorsselmanndeHeer1833}. See also \cite{Dutka:1984uq}.} via the substitution $u=(y/x)^{2}$,
\begin{gather*} \int_{0}^{x} \frac{y^{n}}{(1+y^{2})^{n}}  \frac{dy}{y} = x^{n}\int_{0}^{1} u^{n/2-1}\big(1+x^{2}u\big)^{-n} du = \frac{x^{n}}{n} \, \pFq{2}{1}{n,n/2}{1+n/2}{-x^{2}}. \end{gather*}
We now use the Euler transformation\footnote{\cite[Sections~9--10 in particular]{Euler:1801fk}.}
\begin{gather*}  \pFq{2}{1}{a,b}{c}{z} = (1-z)^{c-a-b}\, \pFq{2}{1}{c-a,c-b}{c}{z} \end{gather*}
to write
\begin{gather*} x\big(1+x^{2}\big)^{n-2} \, \pFq{2}{1}{n/2,n}{1+n/2}{-x^{2}} = \frac{x}{1+x^{2}}\, \pFq{2}{1}{1,1-n/2}{1+n/2}{-x^{2}}. \end{gather*}

We can now use the contiguous relationship between $F$, $F(a^{-})$, and $F(b^{-})$ \cite[Section~7, equation~(7)]{Gauss:1813fk}
\begin{gather}\label{eq:ctgsrel}
	(b-a)(1-z)F = (c-a)F(a^{-})-(c-b)F(b^{-}),
\end{gather}
where $F$ is shorthand for ${}_{2}F_{1}(a,b,c;z)$ and $F(a^{+})$ for ${}_{2}F_{1}(a+1,b,c;z)$ etc., to write this as a~linear combination of a polynomial and $x/(1+x^{2})$. Here, $a=1$, $b=2-n/2$ and $c=1+n/2$, so
\begin{gather*} (1-n/2)\big(1+x^{2}\big) \, \pFq{2}{1}{1,2-n/2}{1+n/2}{-x^{2}} = n/2 - (n-1) \, \pFq{2}{1}{1,1-n/2}{1+n/2}{-x^{2}}; \end{gather*}
therefore
\begin{gather*} \frac{x}{n}\frac{1}{1+x^{2}}\, \pFq{2}{1}{1,1-n/2}{1+n/2}{-x^{2}} = \frac{n-2}{n(n-1)} x \, \pFq{2}{1}{1,2-n/2}{1+n/2}{-x^{2}} + \frac{1}{2(n-1)}\frac{x}{1+x^{2}}. \end{gather*}
We can now integrate this directly to immediately obtain the general form of the Green's function,
\begin{align*}
	G_{n}(\theta) &= 4\frac{R^{2-n}}{S_{n}}\int_{0}^{r} \left( \frac{n-2}{n(n-1)} x \, \pFq{2}{1}{1,2-n/2}{1+n/2}{-x^{2}} + \frac{1}{2(n-1)}\frac{x}{1+x^{2}} \right)  dx \\
	&= 4\frac{R^{2-n}}{S_{n}}\left[ \frac{n-2}{2n(n-1)} x^{2} \, \pFq{3}{2}{1,1,2-n/2}{2,1+n/2}{-x^{2}} + \frac{1}{4(n-1)} \log{(1+x^{2})} \right]_{x=0}^{r} \\
	&= \frac{R^{2-n}}{S_{n}} \left( \frac{n-2}{n(n-1)} r^{2} \, \pFq{3}{2}{1,1,2-n/2}{2,1+n/2}{-r^{2}} + \frac{1}{n-1} \log{(1+r^{2})} \right),
\end{align*}
as required.
\end{proof}

\begin{Corollary} Suppose $n=2m \geq 2$ is a positive even integer. Then the Green's function for the $n$-dimensional sphere of radius $R$ is
\begin{gather}\label{eq:evengfformula}
G_{2m}(\theta) = \frac{R^{2-2m}}{S_{2m}} \left( \frac{m-1}{m(2m-1)} r^{2} \, \pFq{3}{2}{1,1,2-m}{2,1+m}{-r^{2}} + \frac{1}{2m-1} \log\big(1+r^{2}\big)\right),
\end{gather}
where $r=\cot{\frac{1}{2}\theta}$.
\end{Corollary}

In particular, notice that since $ 2-m $ is a negative integer, the hypergeometric term is explicitly a polynomial.

To verify the local behaviour, we extract from \eqref{eq:evengfformula} the most singular term as $1/r \to 0$. Firstly, we have $ r = \cot{\frac{1}{2}\theta} \sim 2/\theta$ as $\theta \to 0$. As discussed in Section~\ref{sec:ccc}, the Green's function should act like the fundamental solution for f\/lat $n$-dimensional space near $1/r=0$, \emph{viz.},
\begin{gather*}
	-\frac{1}{2\pi}\log{s}, \quad  n=2, \qquad \frac{1}{(n-2)S_{n-1}}s^{2-n}, \quad  n\neq 2,
\end{gather*}
where $S_{n-1}=2\pi^{n/2}/\Gamma(n/2)$ is the volume of the $(n-1)$-dimensional unit hypersphere, and we recall that $s$ is the arc-length, given by $s=\abs{x}$ in $ {\mathbb R}^{n} $, and $ s=R\theta $ on hyperspheres~$S^{n}$.

For $ S^{2} $ (i.e., $m=1$), we f\/ind that
\begin{gather*} G_{2} \sim \frac{-2}{4\pi (2\cdot 1-1)} \log{r^{-1}} \sim -\frac{1}{2\pi}\log{s}, \end{gather*}
which is consistent with Section~\ref{sec:ccc}.

For $m \geq 2$, the most singular term is the largest power of $r$ in the hypergeometric polynomial, which has degree $m-2$, and can be simplif\/ied as follows:
\begin{gather*}
\frac{R^{2-2m}}{S_{2m}} \frac{m-1}{m(2m-1)} \frac{(1)_{m-2}(1)_{m-2}(2-m)_{m-2}}{(2)_{m-2}(1+m)_{m-2}}\frac{(-1)^{m}r^{2(m-2)-2}}{(m-2)!} \\
\qquad{} = \frac{\Gamma(m+1/2)}{2\pi^{m+1/2}}\frac{2 \Gamma(m-1) \Gamma(m+1)}{\Gamma(2m+1)} \left(\frac{r}{R}\right)^{2m-2} \\
\qquad{} \sim \frac{\Gamma(m-1)}{\pi^{m}} \frac{m\Gamma(m)\Gamma(m+1/2)}{\sqrt{\pi}2m\Gamma(2m)} \frac{2^{2m-2}}{(R\theta)^{2m-2}} \\
\qquad{} = \frac{\Gamma(m-1)}{\pi^{m}} \frac{1}{2^{2m}} \frac{2^{2m-2}}{(R\theta)^{2m-2}}
= \frac{\Gamma(m)}{(2m-2)2\pi^{m}} s^{2-2m} = \frac{1}{(n-2)S_{n-1}}s^{2-n},
\end{gather*}
using the duplication formula for the $\Gamma$-function; this is also in agreement with Section~\ref{sec:ccc}.

\begin{Remark} We can verify that some simple cases are in agreement with the corresponding results in~\cite{Szmytkowski:2007gf}. Their notation dif\/fers somewhat from ours, and they have the opposite sign convention for the Laplacian (i.e.,~$\Delta$ instead of~$-\Delta$). Specif\/ically, their even-dimensional Green's function for the Laplace case (corresponding to $L=0$ in their paper) is given in their paper's equation~(4.41) by
\begin{gather}
\bar{G}_{0}^{(2n+2)}(\bn,\bn') = \frac{1}{(2n+1)S_{2n+2}} \log{\frac{1-\bn \cdot \bn'}{2}} \notag \\
\hphantom{\bar{G}_{0}^{(2n+2)}(\bn,\bn') =}{} -\frac{n!}{(2n+1)!! S_{2n+2}} \sum_{k=1}^{n} \frac{(2n-2k-1)!!}{k(n-k)!} \frac{C_{k}^{(n-k+1/2)}(\bn \cdot \bn')}{(1-\bn\cdot\bn')^{k}} + \text{const},\label{eq:szmytkowskievengf}
\end{gather}
where $ n $ is our $m-1$, and
\begin{gather*}
C_{\lambda}^{(\alpha)}(x) = \frac{\Gamma(\lambda+2\alpha)}{\Gamma(\lambda+1)\Gamma(2\alpha)} \pFq{2}{1}{-\lambda,\lambda+2\alpha}{\alpha+1/2}{\frac{1-x}{2}}
\end{gather*}
	is the Gegenbauer function (of the f\/irst kind), which for $\lambda \in {\mathbb Z}$ is obviously a polynomial; we have also used that $ C_{0}^{(\alpha)}(x)=1 $ for $ \alpha \neq 0 $ to simplify their expression into~\eqref{eq:szmytkowskievengf}. It is easy to see that $ \bn \cdot \bn' = \cos{\theta} $ in our notation. Putting $n=0$ in~\eqref{eq:szmytkowskievengf} gives the two-dimensional case,
\begin{gather}\label{eq:szmytkowski2}
\bar{G}_{0}^{(2)} = \frac{1}{S_{2}}\log{\frac{1-\cos{\theta}}{2}} + \text{const},
\end{gather}
	whereas if we put $m=1$ in our expression \eqref{eq:evengfformula}, we obtain
\begin{gather*} G_{2}(\theta) = \frac{1}{4\pi}\log{(1+\cot^{2}{\tfrac{1}{2}\theta})} = \frac{1}{2\pi}\log{\csc{\tfrac{1}{2}\theta}} = -\frac{1}{4\pi} \log{\frac{1-\cos{\theta}}{2}}, \end{gather*}
which agrees with \eqref{eq:szmytkowski2} when we take the opposite sign (to agree with the dif\/ference in the Laplacians).
	
Similarly, taking $n=1$ in \eqref{eq:szmytkowskievengf} gives the four-dimensional case,
\begin{align*}
		\bar{G}_{0}^{(4)}(\bn,\bn') &= \frac{1}{3S_{4}} \log{\frac{1-\cos{\theta}}{2}} -\frac{1}{3 S_{4}} \frac{1}{1} \frac{C_{1}^{(1/2)}(\cos{\theta})}{1-\cos{\theta}} + \text{const} \\
		&= \frac{1}{3S_{4}} \left( \log{\frac{1-\cos{\theta}}{2}}-\frac{\cos{\theta}}{1-\cos{\theta}} \right) + \text{const} \\
		&= -\frac{1}{3S_{4}} \left( -\log{\frac{1-\cos{\theta}}{2}}+\frac{1}{2}\cot^{2}{\tfrac{1}{2}\theta} \right) + \text{const},
\end{align*}
	where the constants dif\/fer between the last two lines (essentially we have used $ 1-\cos{\theta}=2\sin^{2}{\frac{1}{2}\theta} $ again). Our corresponding result is
	\begin{align*}
		G_{4}(\theta) = \frac{1}{3S_{4}} \left( -\log{\frac{1-\cos{\theta}}{2}}+ \frac{1}{2}\cot^{2}{\tfrac{1}{2}\theta} \right),
	\end{align*}
	which is once again in agreement when we take the opposite sign convention.
\end{Remark}

\subsubsection[$n$ odd: a recurrence relation]{$\boldsymbol{n}$ odd: a recurrence relation}\label{sec:nodd}

Unfortunately the simple hypergeometric polynomial formula found in the previous section does not extend to the odd case. Instead the Green's function has a quite dif\/ferent character, as we shall see in Theorem~\ref{thm:odddimlgfformula}.

The discussion will be carried out for $\pi/2<\theta<\pi$ initially, then the end result will be obtained by analytic continuation.

\begin{Lemma}[reduction of $J_{2m+1}$ to hypergeometric form]
	Let $\pi/2<\theta<\pi$. With $ J_{n}(\theta)$ as defined in~\eqref{eq:ImJmdefs}, we have
	\begin{gather*} J_{2m+1}(\theta) = \frac{1}{4m} \left( \log{(1+T)} - \frac{T}{2m+1} \, \pFq{3}{2}{1,1,3/2}{2,m+3/2}{-T} \right), \end{gather*}
	where $ T=\tan^{2}{\theta} $.
\end{Lemma}

\begin{proof}
We shall make the substitution
\begin{gather*} u = \frac{\tan^{2}{\psi}}{\tan^{2}{\phi}}, \qquad d\psi = \frac{du}{2u^{1/2}\tan{\phi}(1+u \tan^{2}{\phi})}, \end{gather*}
in the def\/initions of $I_{2m+1}$ and $ J_{2m+1} $ from \eqref{eq:ImJmdefs}. This also gives
\begin{gather*} \sin{\psi} = u^{1/2}\tan{\phi}\big(1+u\tan^{2}{\phi}\big)^{-1/2}. \end{gather*}
Then, returning to the notation of Section~\ref{sec:rr} and taking a general odd $ n=2m+1 $,
\begin{gather*} I_{2m+1}(\phi) = \int_{\phi}^{\pi} \sin^{2m}{\psi} \, d\psi = -\frac{1}{2}\tan^{2m+1}{\phi} \int_{0}^{1} u^{m-1/2}\big(1+u\tan^{2}{\phi}\big)^{-m-1/2} \, du, \end{gather*}
and the usual Euler integral produces
\begin{gather*} I_{2m+1}(\phi) = -\frac{\tan^{2m+1}{\phi}}{2m+1} \, \pFq{2}{1}{m+1/2,m+1}{m+3/2}{-\tan^{2}{\phi}}, \end{gather*}
which can be converted to
\begin{align*}
	I_{2m+1}(\phi) &= -\frac{1}{2m+1} \tan{\phi}\frac{\sin^{2m}{\phi}}{\cos^{2m}{\phi}} \frac{1}{(1+\tan^{2}{\phi})^{m}} \, \pFq{2}{1}{1,1/2}{m+3/2}{-\tan^{2}{\phi}} \\
	&= -\frac{1}{2m+1} \tan{\phi}\sin^{2m}{\phi} \, \pFq{2}{1}{1,1/2}{m+3/2}{-\tan^{2}{\phi}}
\end{align*}
using an Euler transformation.

We then have
\begin{gather*}
	J_{2m+1}(\theta) = -\frac{1}{2m+1}\int_{\theta}^{\pi} \tan{\phi} \, \pFq{2}{1}{1,1/2}{m+3/2}{-\tan^{2}{\phi}} \, d\phi.
\end{gather*}
Substituting
\begin{gather*} v = \frac{\tan^{2}{\phi}}{\tan^{2}{\theta}}=:T^{-1}\tan^{2}{\phi}, \qquad \tan{\phi} \, d\phi = \frac{dv}{2(1+Tv)}, \end{gather*}
the integral becomes
\begin{gather*} J_{m}(\theta) = \frac{1}{2(2m+1)}\int_{0}^{1} \frac{T}{1+Tv} \, \pFq{2}{1}{1,1/2}{m+3/2}{-Tv} \, dv. \end{gather*}
The same contiguous relation as before (see \eqref{eq:ctgsrel} above) allows us to write this as
\begin{gather*} J_{2m+1}(\theta) = \frac{T}{4m}\int_{0}^{1} \left( \frac{1}{1+Tv} - \frac{1}{2m+1} \, \pFq{2}{1}{1,3/2}{m+3/2}{-Tv} \right) dv, \end{gather*}
which integrates immediately to
\begin{gather*} J_{2m+1}(\theta) = \frac{1}{4m} \left( \log{(1+T)} - \frac{T}{2m+1} \, \pFq{3}{2}{1,1,3/2}{2,m+3/2}{-T} \right), \end{gather*}
as required.
\end{proof}

This lemma provides a form of the Green's function, but rather an unilluminating one: indeed, it is not even obvious that there is a f\/inite form to the expansion. We shall now derive the following explicitly f\/inite form.

\begin{Theorem}\label{thm:odddimlgfformula}
Suppose that $n=2m+1 \geq 3$ is a positive odd integer. Then the Green's function for the sphere of dimension $2m+1$ is
\begin{gather*}
G_{2m+1}(\theta) = \frac{R^{1-2m}}{S_{2m+1}}\frac{1}{2m} \sum_{k=0}^{m-1} \left( \binom{k-1/2}{k}y^{k} \left( 1+(\pi-\theta)\cot{\theta} \right) - \frac{1}{3}\sum_{l=1}^{k} \frac{\binom{k-1/2}{k}}{\binom{l+1/2}{l-1}}y^{k-l} \right),
	\end{gather*}
	where $y=\csc^{2}{\theta}$.
\end{Theorem}

Hence the solution can be expressed as an even polynomial in $ \cot{\theta} $, added to an odd polynomial in $ \cot{\theta} $ multiplied by $ \pi-\theta $.

\begin{proof}
We shall employ an argument using reduction formulae, based on the contiguous relations for hypergeometric functions.

We consider the hypergeometric expression in $ J_{2m+1} $ derived in the previous lemma: write
\begin{gather*}
	A(m) := -\frac{T}{2m+1} \, \pFq{3}{2}{1,1,3/2}{2,m+3/2}{-T}.
\end{gather*}
By the ${}_{3}F_{2}$ contiguous relation
\begin{gather*} (b-a-1)F = (b-1)F(b^{-})-aF(a^{+}), \end{gather*}
where $ F=F(a,\dotsc;b,\dotsc;z) $ \cite[equation~(15)]{Rainville:1945uq},
we know $A$ satisf\/ies
\begin{gather*}
	A(m)-A(m-1) = \frac{2T}{(2m+1)(2m-1)} \pFq{2}{1}{1,3/2}{m+3/2}{-T}.
\end{gather*}

Hence we can write down $A(m)$ as a sum of hypergeometric functions:
\begin{gather}\label{eq:1strecur}
	A(m) = A(0) + \sum_{k=0}^{m-1} \frac{2T}{(2k+3)(2k+1)} \pFq{2}{1}{1,3/2}{k+5/2}{-T}
\end{gather}
and we also f\/ind that
\begin{gather*}
	A(0) = -\frac{1}{2}\log{(1+T)},
\end{gather*}
which cancels the other logarithm term, conf\/irming our suspicion of the initial solution.
This is still not obviously f\/initary. Now consider the term in the sum
\begin{gather*}
	B(k) := \frac{2T}{(2k+3)(2k+1)} \, \pFq{2}{1}{1,3/2}{k+5/2}{-T}.
\end{gather*}
We have the other contiguous relation for a hypergeometric function~\cite[equation~(9)]{Gauss:1813fk}
\begin{gather*} (1-a + (c-b-1)z)F = (c-a)F(a^{-}) -(c-1)(1-z)F(c^{-}), \end{gather*}
so after rearranging, $ B(k) $ satisf\/ies
\begin{gather}\label{eq:2ndrecur}
	B(k) = -\frac{2}{(2k+1)(2k)}+ y\frac{2k-1}{2k} B(k-1),
\end{gather}
where we have written $ y=1+1/T=\csc^{2}(\theta) $. We also have
\begin{gather*}
	B(0) = 2\left( 1-\frac{\arctan{\sqrt{T}}}{\sqrt{T}} \right).
\end{gather*}
Now split $B(k)$ as $ B(k)=\alpha(k)\beta(k) $, where $ \alpha(0)=B(0)$, and $\alpha(k)$ solves the homogeneous version of \eqref{eq:2ndrecur}:
\begin{gather*}
	\alpha(k) = \frac{2k+1}{2k}y \alpha(k-1).
\end{gather*}
We can show that
\begin{gather*}
	\alpha(k) = y^{k}\binom{k-1/2}{k} B(0)
\end{gather*}
(for example, by checking the recurrence relation and initial conditions). Now, $\beta$ satisf\/ies the recurrence relation
\begin{gather*}
	\alpha(k)\beta(k) = -\frac{2}{(2k+1)(2k)}+y\frac{2k-1}{2k}\alpha(k-1)\beta(k-1) = -\frac{2}{(2k+1)(2k)}+\alpha(k)\beta(k-1),
\end{gather*}
and hence we obtain the simple recurrence for $\beta$:
\begin{gather*}
	\beta(k)-\beta(k-1)=-\frac{2}{(2k+1)(2k)\alpha(k)}.
\end{gather*}
Then, since $\beta(0)=1$, we obtain
\begin{gather*}
	\beta(k) = 1-\sum_{l=1}^{k} \frac{2}{(2l+1)(2l)\alpha(l)},
\end{gather*}
and since
\begin{gather*}
	(2l+1)(2l)\binom{l-1/2}{l} = 3\binom{l+1/2}{l-1},
\end{gather*}
we f\/ind that
\begin{gather*}
	B(k) = \alpha(k)-\!\sum_{l=1}^{k}\! \frac{2\alpha(k)}{(2l+1)(2l)\alpha(l)} = 2\binom{k\!-\!1/2}{k} y^{k} \!\left( \! 1-\frac{\arctan{\sqrt{T}}}{\sqrt{T}} \!\right)\! - \frac{2}{3}\!\sum_{l=1}^{k}\! \frac{\binom{k-1/2}{k}}{\binom{l+1/2}{l-1}}y^{k-l} .
\end{gather*}
Now we can solve \eqref{eq:1strecur} to obtain
\begin{gather*}
	A(m) = -\frac{1}{2}\log{(1+T)} + 2\sum_{k=0}^{m-1} \left( \binom{k-1/2}{k}y^{k} \left( 1-\frac{\arctan{\sqrt{T}}}{\sqrt{T}} \right) - \frac{1}{3}\sum_{l=1}^{k} \frac{\binom{k-1/2}{k}}{\binom{l+1/2}{l-1}}y^{k-l} \right).
\end{gather*}

One can see from the expansion near $ \theta=\pi $ that the correct branch of $\arctan{(\tan{\theta})}$ to choose is $ \theta-\pi $, so we can substitute back to f\/ind that
\begin{gather*}
	J_{2m+1}(\theta) = \frac{1}{2m} \sum_{k=0}^{m-1} \left( \binom{k-1/2}{k}y^{k} \left( 1+(\pi-\theta)\cot{\theta} \right) - \frac{1}{3}\sum_{l=1}^{k} \frac{\binom{k-1/2}{k}}{\binom{l+1/2}{l-1}}y^{k-l} \right).
\end{gather*}

Hence,
\begin{gather*}
	G_{2m+1}(\theta) = \frac{R^{1-2m}}{S_{2m+1}}\frac{1}{2m} \sum_{k=0}^{m-1} \left( \binom{k-1/2}{k}y^{k} \left( 1+(\pi-\theta)\cot{\theta} \right) - \frac{1}{3}\sum_{l=1}^{k} \frac{\binom{k-1/2}{k}}{\binom{l+1/2}{l-1}}y^{k-l} \right),
\end{gather*}
as required.
\end{proof}

The most singular term in $G_{2m+1}(\theta)$ is
\begin{gather*}
	\frac{R^{1-2m}}{S_{2m+1}}\frac{1}{2m} \binom{m-3/2}{m-1} \pi\csc^{2m-2}{\theta}\cot{\theta} \sim \frac{1}{S_{2m+1}}\frac{1}{2m} \binom{m-3/2}{m-1} \pi \frac{1}{(R\theta)^{2m-1}},
\end{gather*}
and then the coef\/f\/icient of $ (R\theta)^{2-(2m+1)} $ is
\begin{align*}
	\frac{1}{S_{2m+1}}\frac{1}{2m} \binom{m-3/2}{m-1} \pi &= \frac{\Gamma(m+1)}{2\pi^{m+1}} \frac{\Gamma(m-1/2)\pi}{2m\Gamma(m)\Gamma(1/2)} \\
	&= \frac{\Gamma(m-1/2)}{4\pi^{m+1/2}} = \frac{\Gamma(m+1/2)}{(2m-1)2\pi^{m+1/2}} = \frac{1}{((2m+1)-2)S_{2m}},
\end{align*}
as it should be to agree with the f\/lat-space results.

\begin{Remark}
As in the even-dimensional case, we can compare this with the result in~\cite{Szmytkowski:2007gf}. The three-dimensional case is given in that paper's equation~(4.23):
\begin{gather*}
\bar{G}_{0}^{(3)}(\bn,\bn') = -\frac{1}{4\pi^{2}} \frac{\bn \cdot \bn'}{\sqrt{1-\bn \cdot \bn'}} \arccos{(-\bn \cdot \bn')} + \text{const}.
\end{gather*}
Since $ \arccos{(-\bn \cdot \bn')}=\pi-\theta $, substituting in gives
\begin{gather}	\label{eq:szmytkowski3gf}
\bar{G}_{0}^{(3)}(\bn,\bn') = -\frac{1}{4\pi^{2}}(\pi-\theta)\frac{\cos{\theta}}{\abs{\sin{\theta}}}+\text{const} = -\frac{1}{4\pi^{2}}(\pi-\theta)\cot{\theta}+\text{const},
\end{gather}
since $ \sin{\theta}\geq 0$ for $0\leq\theta\leq \pi$.
	
	Whereas, if we put $m=1$ (and $R=1$) in our result we f\/ind that
\begin{gather*}
G_{3}(\theta) = \frac{1}{2S_{3}} \binom{-1/2}{0} y^{0} (1+(\pi-\theta)\cot{\theta}) - 0 = \frac{1}{4\pi^{2}} \left( (\pi-\theta)\cot{\theta}+1 \right),
\end{gather*}
which agrees with \eqref{eq:szmytkowski3gf} when the sign convention dif\/ference is understood.
	
A similar calculation using \cite[equation~(4.34)]{Szmytkowski:2007gf} shows that the f\/ive-dimensional results also agree up to a constant, ours being
\begin{gather*}
G_{5}(\theta) = \frac{1}{S_{5}} \left( \frac{1}{8}\big(2+\csc^{2}{\theta}\big)(\pi-\theta)\cot{\theta} + \frac{1}{24}\big(5+3\csc^{2}{\theta}\big) \right).
\end{gather*}
\end{Remark}

\section{Applications}

\subsection{The dipole potential on a sphere}

In Euclidean space, the conventional way to def\/ine a dipole potential from a fundamental solution is by considering two point charges with charge $q$ and $-q$ at $a$ and $P(d)$ a distance $d$ apart. We then take the limit $ q\to \infty $, $ d \downarrow 0$, so that $qd$ is constant, and the path has a well-def\/ined tangent vector $p$ when $d$ becomes zero. To understand this on the sphere, it is advantageous to make this a little more formal:

\begin{Definition}
Let $\gamma\colon (-\varepsilon,\varepsilon) \to M $ be a path on the manifold $M$, with $ \gamma(0)=a$ and $ \dot{\gamma}(0) = p \in T_{a}(M) $. Let $ G(x,y)$ be the Green's function for the Laplace--Beltrami operator on~$M$. The \emph{dipole potential at~$a$ with dipole moment~$p$} is given by
\begin{gather*}
H(x,a,p) = \lim_{t \to 0} \frac{G(x,\gamma(t))-G(x,a)}{t} = p \cdot \nabla_{a} G(x,a).
\end{gather*}
\end{Definition}

The limit def\/inition indicates that this describes two point charges of charge $ \pm t^{-1}$ being brought together; the second equality shows that $H$ is well-def\/ined by prescribing~$a$ and~$p$. Moreover, since the total charge is always zero, the dipole potential is also consistent on compact manifolds, such as the hypersphere. Since two fundamental solutions are subtracted from one another, we also lose any dependence on the constant we took in the def\/inition of our particular fundamental solution, so this is in a sense \emph{the} dipole potential.

We therefore need to f\/ind the gradient of $G(x,a)$. Since this involves moving the base point of the coordinates, and hence dif\/ferent def\/initions of the angle $\theta$, it is advantageous to start by recalling that on the sphere of radius~$R$, $\theta$ can actually be def\/ined as a particular solution to the equation
\begin{gather*}
	\cos{\theta} = \frac{x \cdot a}{R^{2}},
\end{gather*}
since it is the angle between the vectors $ x,a \in {\mathbb R}^{n+1} $. We then f\/ind that
\begin{gather*}
	\nabla_{a} \theta = \frac{1}{R^{2}\sin{\theta}} x,
\end{gather*}
and so the chain rule gives
\begin{gather*}
	p \cdot \nabla_{a} G(x,a) = G_{n}'(\theta) (p \cdot \nabla_{a}) \theta = \frac{G_{n}'(\theta)}{R^{2}\sin{\theta}} (p \cdot x),
\end{gather*}
and using the group of rotations on the hypersphere, we can always use the spherical symmetry to place $a$, $ p $ and $x$ in the same three-dimensional space, with $ a $ along the axis $ \theta=0 $ and $ p $ in the plane $ \theta=\phi=0 $, so that $ p\cdot x = R\abs{p} \sin{\theta}\cos{\phi} $, and then
\begin{gather*}
	p \cdot \nabla_{a} G(x,a) = \frac{G_{n}'(\theta)}{R} \abs{p} \cos{\phi}.
\end{gather*}

We should now verify that this function has the correct behaviour as $ R \to \infty $, in that it becomes a dipole at the origin in the Euclidean space with distance $ s=R\theta $,
\begin{align*}
	\frac{1}{R}G_{n}'(\theta) \abs{p} \cos{\phi} &= \frac{ \abs{p} \cos{\phi}}{S_{n}R^{n-1}}\csc^{n-1}(\theta) \int_{\theta}^{\pi} \sin^{n-1}{\theta'} \, d\theta' \\
	&\sim \frac{ \abs{p} \cos{\phi}}{S_{n}(R\theta)^{n-1}} \frac{\pi^{1/2}\Gamma(n/2)}{\Gamma((n+1)/2)} = \frac{ s^{1-n}\abs{p} \cos{\phi}}{S_{n-1}},
\end{align*}
similarly to the calculation in Section~\ref{sec:ccc}. It is easy to check that the dipole potential at the origin in $ {\mathbb R}^{n} $ is given by
\begin{gather*}
	\left. (p \cdot \nabla_{a}) \frac{\abs{x-a}^{2-n}}{(n-2)S_{n-1}} \right\rvert_{a=0} = \frac{\abs{x}^{-n} p \cdot x}{S_{n-1}} = \frac{s^{1-n} \abs{p} \cos{\phi}}{S_{n-1}},
\end{gather*}
with which the limit agrees.

\subsection[Fourier expansion of our fundamental solution on $S^{2}$]{Fourier expansion of our fundamental solution on $\boldsymbol{S^{2}}$}

The form of the $ n=2 $ result means that we can quite easily derive an azimuthal Fourier expansion. We follow~\cite{Cohl:2015rp}, with some simplif\/ications. We have
\begin{gather*}
	\cos{d(x,x')} = \cos{\theta}\cos{\theta'}+\sin{\theta}\sin{\theta'}\cos{(\phi-\phi')},
\end{gather*}
with $x \!=\! (\cos{\phi}\sin{\theta},\sin{\phi}\sin{\theta},\cos{\theta})$ in standard spherical coordinates (and similarly $x'$), $d(x,x')$ the spherical distance between them, and inserting this into our fundamental solution gives
\begin{gather*}
G_{2}(x,x') = -\frac{1}{4\pi} \log{\frac{1}{2}\left( 1-\cos{\theta}\cos{\theta'}-\sin{\theta}\sin{\theta'} \cos{(\phi-\phi')} \right)}.
\end{gather*}
We can f\/ind the expansion of this quite easily using the following:

\begin{Lemma}Let $ A>B>0 $ and $ t \in {\mathbb R}$. Then
\begin{gather}\label{eq:logABcos}
-\log{(A+B\cos{t})} = \log{\frac{2\big(A-\sqrt{A^{2}-B^{2}}\big)}{B^{2}}} + 2 \sum_{k=1}^{\infty} \frac{\cos{kt}}{k} \left( \frac{\sqrt{A^{2}-B^{2}}-A}{B} \right)^{k}.
\end{gather}
\end{Lemma}

\begin{proof}
We begin with the well-known formula
\begin{gather*}
-\log{(1+r^{2}-2r\cos{t})} = 2\sum_{k=1}^{\infty} \frac{r^{k}\cos{kt}}{k}, \qquad \abs{r}<1,
\end{gather*}
which can by proven by applying Euler's formula to the right-hand side and computing the sum directly. We then have
\begin{gather*}
-\log{\left(1-\frac{2r}{1+r^{2}}\cos{t} \right)} = \log\big(1+r^{2}\big) + 2\sum_{k=1}^{\infty} \frac{r^{k}\cos{kt}}{k},
\end{gather*}
and it is easy to see that taking
\begin{gather}\label{eq:requationBA}
\frac{2r}{1+r^{2}} = -\frac{B}{A},
\end{gather}
the left-hand side becomes
\begin{gather*}
-\log{\left(1+\frac{B}{A}\cos{t} \right)} = \log{A}-\log{(A+B\cos{t})}.
\end{gather*}
To f\/ind the right-hand side, we need to f\/ind the correct value of $r$ in \eqref{eq:requationBA}, but since $ \abs{r}<1 $, the only choice possible is
\begin{gather*}
r = \frac{\sqrt{A^{2}-B^{2}}-A}{B},
\end{gather*}
and some brief algebra gives the result.
\end{proof}

\begin{Theorem}
	Let $ \phi,\phi' \in (0,2\pi) $ and $ \theta,\theta' \in (0,\pi) $. Then the azimuthal Fourier expansion for $ G_{2}(x,x') $ in spherical coordinates is given by
\begin{gather}\label{eq:2dgffourier}
2\pi G_{2}(x,x') = -\frac{1}{2}\log{\frac{(1-m)(1+M)}{4}} + \sum_{k=1}^{\infty} \frac{\cos{k(\phi-\phi')}}{k} \left( \frac{(1+m)(1-M)}{(1+M)(1-m)} \right)^{k/2},
\end{gather}
where $ M=\max{\{\cos{\theta},\cos{\theta'}\}} $ and $ m=\min{\{\cos{\theta},\cos{\theta'}\}} $.
\end{Theorem}

\begin{proof}
We can apply this directly, with $ 2A = 1-\cos{\theta}\cos{\theta'} $ and $ 2B=-\sin{\theta}\sin{\theta'} $. Then
\begin{gather*}
\frac{\sqrt{A^{2}-B^{2}}-A}{B} = \frac{\abs{\cos{\theta}-\cos{\theta'}}-1+\cos{\theta}	\cos{\theta'}}{\sin{\theta}\sin{\theta'}},
\end{gather*}
and squaring, square-rooting and applying trigonometric identities, this in turn becomes
\begin{gather*}
\sqrt{\frac{(1+m)(1-M)}{(1+M)(1-m)}},
\end{gather*}
where $ M $ is the larger of $ \cos{\theta}$, $\cos{\theta'} $ and $m$ the smaller. We hence obtain the result.
\end{proof}

We can write \eqref{eq:2dgffourier} most succinctly as
\begin{gather*}
		2\pi G_{2}(x,x') = \log{\left( \csc\big(\tfrac{1}{2}\theta_{>}\big)\sec\big(\tfrac{1}{2}\theta_{<}\big) \right)} + \sum_{k=1}^{\infty} \frac{\cos{k(\phi-\phi')}}{k} \left( \frac{\cot{\tfrac{1}{2}\theta_{>}}}{\cot{\tfrac{1}{2}\theta_{<}}} \right)^{k},
\end{gather*}
where $ \theta_{>} $ is the larger of $ \theta$, $\theta'$, and $ \theta_{<} $ the smaller. Taking the limit as $R \to \infty$ with $ s=R\theta $ constant gives
\begin{align*}
	2\pi G_{2}(x,x')-\log{2R} &\to -\log{(R\theta_{>})} + \sum_{k=1}^{\infty} \frac{\cos{k(\phi-\phi')}}{k} \left( \frac{R\theta_{<}}{R\theta_{>}} \right)^{k} \\
	& = -\log{s_{>}} + \sum_{k=1}^{\infty} \frac{\cos{k(\phi-\phi')}}{k} \left( \frac{s_{<}}{s_{>}} \right)^{k},
\end{align*}
which agrees with the f\/lat-space result quoted in~\cite{Cohl:2015rp}, which we can also derive from~\eqref{eq:logABcos}\footnote{We had to subtract an extra constant due to our convention that $G_{2}$ is positive, which cannot happen in f\/lat space as $\log$ is not bounded, above or below.}.

\subsection[A fundamental solution on ${\mathbb R} \mathbb{P}^{n}$]{A fundamental solution on $\boldsymbol{{\mathbb R} \mathbb{P}^{n}}$}
The real projective space ${\mathbb R} \mathbb{P}^{n}$ can be considered as the quotient of the sphere $S^{n}$ by the group $\{ 1, -1 \}$, where $-1$ is the antipodal map $x\mapsto -x$, isomorphic to $Z_{2}$: ${\mathbb R} \mathbb{P}^{n} \cong S^{n}/Z_{2}$. This is realised in our spherical coordinates by associating the points
\begin{gather*} (\theta_{i},\varphi) \sim (\pi-\theta_{i},\pm\pi+\varphi). \end{gather*}
Therefore our fundamental solution for $S^{n}$ can be made into one for ${\mathbb R} \mathbb{P}^{n}$ by setting
\begin{gather*} G_{{\mathbb R} \mathbb{P}^{n}}(x,x') = \frac{1}{2}(G_{S^{n}}(x,x') + G_{S^{n}}(-x,x')), \end{gather*}
or in terms of angles,
\begin{gather*} G_{{\mathbb R} \mathbb{P}^{n}}(\theta_{1}) = \frac{1}{2}(G_{S^{n}}(\theta_{1}) + G_{S^{n}}(\pi-\theta_{1})). \end{gather*}
This clearly has all the correct properties of a fundamental solution, such as matching singularities and smoothness.

\subsection{Interpretation of prior results}

In this section we shall discuss the results of \cite{Cohl:2011vn}, and in particular, how our construction enables a consistent explanation of some of the properties of their solution. We shall apply a similar procedure to the previous section, but consider the functions odd under the group action.

Consider f\/irst the equation
\begin{gather}\label{eq:cohllapeq}
	-\Delta{G(x,x')} = \delta(x,x') - \delta(-x,x');
\end{gather}
this is consistent on the sphere, because the integral of the right-hand side is zero. It is essentially the opposite limit of the two point charge case to the dipole considered above: the separation of the charges is now the largest possible. Linearity allows us to conclude that
\begin{gather}\label{eq:cohlgfdef}
	\tilde{G}(x,x') := G(x,x') - G(-x,x')
\end{gather}
is a solution to \eqref{eq:cohllapeq}, where $ G $ is the generalised Green's function discussed in Section~\ref{sec:ggf}, since
\begin{gather*}
	-\Delta{\tilde{G}(x,x')} = \big(\delta(x,x')-(S_{n})^{-1}\big)-\big(\delta(-x,x')-(S_{n})^{-1}\big) = \delta(x,x') - \delta(-x,x').
\end{gather*}

In fact, \eqref{eq:cohllapeq} is the equation that the Green's function found in \cite{Cohl:2011vn} satisf\/ies: to see this we only have to observe its oddness under the transformation $x \mapsto -x$, or in our terms, $\theta \mapsto \pi-\theta$. This is readily apparent if we consider
the form~\cite[equation~(29)]{Cohl:2011vn}
\begin{gather*}
	\cos{\theta} \, \pFq{2}{1}{1/2,n/2}{3/2}{\cos^{2}{\theta}}.
\end{gather*}
Hence it must have a singularity at $\theta=\pi$ of the same nature and opposite sign as it does at $\theta=0$, where it was designed to have a~singularity that produces a $\delta$-function, so it must satisfy~\eqref{eq:cohllapeq}.

An instructive and simple example is given by the nontrivial case in two dimensions: we have
\begin{align*}
G(x,x')-G(-x,x') &= \frac{1}{4\pi}\log{\csc^{2}{\tfrac{1}{2}\theta}} - \frac{1}{4\pi}\log{\csc^{2}{\tfrac{1}{2}(\pi-\theta)}} \\
&= -\frac{1}{4\pi}\log{\sin^{2}{\tfrac{1}{2}\theta}} + \frac{1}{4\pi}\log{\cos^{2}{\tfrac{1}{2}\theta}} = \frac{1}{4\pi}\log{\cot^{2}{\tfrac{1}{2}\theta}},
\end{align*}
which is precisely the dimension-two Green's function calculated in \cite{Cohl:2011vn}. It also follows imme\-diately that this satisf\/ies~\eqref{eq:cohllapeq}.

We now proceed to prove the relationship in the general case: we recall that our integral, from \eqref{eq:ImJmdefs}, is given by
\begin{gather*}
	J_{n}(\theta) = \int_{\theta}^{\pi} \csc^{n-1}{\phi} \int_{\phi}^{\pi} \sin^{n-1}{\psi} \, d\psi \, d\phi,
\end{gather*}
and Cohl's is
\begin{gather*}
	\mathcal{I}_{n}(\theta) := \int_{\theta}^{\pi/2} \frac{dx}{\sin^{n-1}{x}},
\end{gather*}
up to a constant.

\begin{Proposition}
For $ \theta \in (0,\pi)$, we have
\begin{gather}\label{eq:JcurlyIrel}
J_{n}(\theta)-J_{n}(\pi-\theta) = \frac{\pi^{1/2}\Gamma(n/2)}{\Gamma((n+1)/2)} \mathcal{I}_{n}(\theta)
\end{gather}
\end{Proposition}

\begin{proof}
We have
\begin{gather*}
	J_{n}(\theta) = \int_{\theta}^{\pi} \csc^{n-1}{\phi} \left( \int_{\pi/2}^{\pi} \sin^{n-1}{\psi} \, d\psi + \int_{\phi}^{\pi/2} \sin^{n-1}{\psi} \, d\psi \right)  d\phi,
\end{gather*}
so the dif\/ference is, using $\int_{a}^{b} = \int_{a}^{c}+\int_{c}^{b}$,
\begin{gather*}
	J_{n}(\theta)-J_{n}(\pi-\theta) = \int_{\theta}^{\pi-\theta} \csc^{n-1}{\phi} \left( \int_{\pi/2}^{\pi} \sin^{n-1}{\psi} \, d\psi + \int_{\phi}^{\pi/2} \sin^{n-1}{\psi} \, d\psi \right) d\phi,
\end{gather*}
valid for $ 0<\theta<\pi $ when we recall the standard def\/inition $ \int_{a}^{b}=-\int_{b}^{a} $ if $ b<a $.
The f\/irst integral in the bracket is the same as
\begin{gather*}
	\int_{0}^{\pi/2} \sin^{n-1}{\psi} \, d\psi = \frac{\sqrt{\pi}\Gamma(n/2)}{2\Gamma((n+1)/2)}
\end{gather*}
using the Beta-function; the second is odd about $\phi=\pi/2$. A simple way to see this is to change variables and notice that $\cos^{n-1}{x}$ is an even function, so it has an odd antiderivative. Since $\csc^{n-1}{\phi}$ is even about $\phi=\pi/2$ (again, a change of variables will show this easily), the second term's integrand is odd overall, and hence the integral over a region symmetric about $\phi=\pi/2$ is zero. We are therefore left with
\begin{gather*}
	\frac{\sqrt{\pi}\Gamma(n/2)}{2\Gamma((n+1)/2)} \int_{\theta}^{\pi-\theta} \csc^{n-1}{\phi} \, d\phi;
\end{gather*}
this time, since the integrand is even, the answer is double the integral from $\theta$ to $\pi/2$, which gives the result.
\end{proof}

With that resolved, we just have to check the constants. In our case,
\begin{gather*}
	G_{n}(\theta) = \frac{\Gamma((n+1)/2)}{2\pi^{(n+1)/2}R^{n-2}} J_{n}(\theta).
\end{gather*}
Using \eqref{eq:JcurlyIrel}, we have
\begin{align*}
G_{n}(\theta)-G_{n}(\pi-\theta) &= \frac{\pi^{1/2}\Gamma(n/2)}{\Gamma((n+1)/2)} (J_{n}(\theta)-J_{n}(\pi-\theta)) \\
&= \frac{\pi^{1/2}\Gamma(n/2)}{\Gamma((n+1)/2)} \frac{\Gamma((n+1)/2)}{2\pi^{(n+1)/2}R^{n-2}} \mathcal{I}_{n}
= \frac{\Gamma(n/2)}{2\pi^{n/2}R^{n-2}} \mathcal{I}_{n},
\end{align*}
which is the same as Cohl's solution. Hence the function in Cohl's paper is related to ours in the way we specif\/ied in \eqref{eq:cohlgfdef}.

\appendix
\section[Characterisation of ${}_{2}F_{1}$ with half-integer coef\/f\/icients]{Characterisation of $\boldsymbol{{}_{2}F_{1}}$ with half-integer coef\/f\/icients}
We can use the contiguous relations for hypergeometric functions to derive a collection of representation results.

Let $ n,n_{1},n_{2},\dotsc $ be arbitrary positive integers, and $ h,h_{1},h_{2},\dots$ be arbitrary proper half-integers, i.e., $ h,h_{i} \in {\mathbb Z}+\frac{1}{2} $ for all $i$. We also adopt an abbreviated notation for the ordinary hypergeometric function,
\begin{gather*} [a, b; c] = \pFq{2}{1}{a,b}{c}{z}. \end{gather*}
Let also ${\mathbb Z}(z)$ be the f\/ield of rational functions in $z$ with integer coef\/f\/icients, and $V(f_{i}(z))$ be the ${\mathbb Z}(z)$-vector space with basis $\{f_{i}(z)\}$. Then we have:
\begin{Theorem}There is the following classification of ${}_{2}F_{1}$ hypergeometric functions with half-integer and integer coefficients:
\begin{enumerate}\itemsep=0pt
\item[$1)$] if $a \in {\mathbb C}$ is arbitrary, and $c \in {\mathbb C} \setminus \{ 0,-1,-2,\dotsc,-n \}$, then $[a,-n;c] \in {\mathbb Q}[z]$ is a polynomial;
\item[$2)$] $[n_{1},n_{2};n_{3}] \in V(1)$, i.e., a rational function of~$z$;
\item[$3)$] $[h,n_{1};n_{2}] \in V\big(1,\sqrt{1-z}\big)$;
\item[$4)$] $[n_{1},n_{2};h] \in V\big(1,(1-z)^{-1/2}z^{-1/2}\arcsin{z^{1/2}}\big)$;
\item[$5)$] $[h_{1},h_{2};n] \in V\big( \frac{2}{\pi} E(z), \frac{2}{\pi}K(z) \big)$;
\item[$6)$] $[h_{1},n;h_{2}] \in V\big( 1, z^{-1/2}\operatorname{arctanh}{z^{1/2}} \big)$\footnote{This notation is illogical, but common: see \cite[p.~80]{Pringsheim:1911ed}.};
\item[$7)$] $[h_{1},h_{2};h_{3}] \in V\big( \sqrt{1-z}, z^{-1/2}\arcsin{z^{1/2}} \big)$,
\end{enumerate}
where $K(z)$ and $E(z)$ are the complete elliptic integrals of the first and second kinds respectively.
\end{Theorem}

This is a ref\/inement and specialisation of Theorem~1.1 of~\cite{Vidunas:2003kx} that is useful in the odd-dimensional case considered in Section~\ref{sec:nodd}; in our classif\/ication, the rational functions are over the integers since the parameters of the hypergeometric functions are rational.

\begin{Remark}
The Euler transformation shows that Cases~4 and~7 are the same.
\end{Remark}

\begin{proof} Case~1 is trivial. The other results follow from the contiguous relations and induction. These allow the expression of any one ${}_{2}F_{1}$ in terms of two others with parameters integer steps away, with rational coef\/f\/icients\footnote{For more detail, see~\cite{Vidunas:2003kx}.}. In particular, we may arrange circumstances so that
\begin{gather*} \pFq{2}{1}{a,b}{c}{z} = P(z) \, \pFq{2}{1}{\alpha,\beta}{\gamma}{z} + Q(z) \, \pFq{2}{1}{\alpha-1,\beta}{\gamma}{z}, \end{gather*}
where $\{\alpha,\beta,\gamma\} \subset \{1/2,1,3/2\}$. In particular, in the cases in the theorem, the forms of $[\alpha,\beta,\gamma]$ are
\begin{enumerate}\itemsep=0pt
		\item[$2)$] $[1,1,1]$;
		\item[$3)$] $[1,1/2,1]$;
		\item[$4)$] $[1,1,3/2]$;
		\item[$5)$] $[1/2,1/2,1]$;
		\item[$6)$] $[1,1/2,3/2]$;
		\item[$7)$] $[3/2,1/2,3/2]$.
\end{enumerate}
Evaluating the power series of the given hypergeometric functions with these parameters gives the result.
\end{proof}

\subsection*{Acknowledgements}

The author would like to thank David Stuart for suggesting a method for solving the odd case recurrence relation, and Thomas Forster for some notational advice, as well as the referees for numerous suggestions which improved the paper. We especially wish to acknowledge our indebtedness to the f\/irst anonymous referee, without whose extraordinary and tireless attention, and extensive suggestions and contributions, so many aspects of the paper would have been considerably the poorer.

\pdfbookmark[1]{References}{ref}
\LastPageEnding


\begin{thebibliography}{99}
\footnotesize\itemsep=0pt

\bibitem{Aubin:1998vn}
Aubin T., Some nonlinear problems in {R}iemannian geometry, \href{http://dx.doi.org/10.1007/978-3-662-13006-3}{\textit{Springer Monographs
 in Mathematics}}, Springer-Verlag, Berlin, 1998.

\bibitem{Cohl:2011vn}
Cohl H.S., Fundamental solution of {L}aplace's equation in hyperspherical
 geometry, \href{http://dx.doi.org/10.3842/SIGMA.2011.108}{\textit{SIGMA}} \textbf{7} (2011), 108, 14~pages,
 \href{http://arxiv.org/abs/1108.3679}{arXiv:1108.3679}.

\bibitem{Cohl:2012ly}
Cohl H.S., Kalnins E.G., Fourier and {G}egenbauer expansions for a fundamental
 solution of the {L}aplacian in the hyperboloid model of hyperbolic geometry,
 \href{http://dx.doi.org/10.1088/1751-8113/45/14/145206}{\textit{J.~Phys.~A: Math. Theor.}} \textbf{45} (2012), 145206, 32~pages,
 \href{http://arxiv.org/abs/1105.0386}{arXiv:1105.0386}.

\bibitem{Cohl:2015rp}
Cohl H.S., Palmer R.M., Fourier and {G}egenbauer expansions for a fundamental
 solution of {L}aplace's equation in hyperspherical geometry, \href{http://dx.doi.org/10.3842/SIGMA.2015.015}{\textit{SIGMA}}
 \textbf{11} (2015), 015, 23~pages, \href{http://arxiv.org/abs/1405.4847}{arXiv:1405.4847}.

\bibitem{Courant:1953zr}
Courant R., Hilbert D., Methods of mathematical physics. {V}ol.~{I},
 Interscience Publishers, Inc., New York, N.Y., 1953.

\bibitem{Dutka:1984uq}
Dutka J., The early history of the hypergeometric function, \href{http://dx.doi.org/10.1007/BF00330241}{\textit{Arch. Hist.
 Exact Sci.}} \textbf{31} (1984), 15--34.

\bibitem{Euler:1801fk}
Euler L., Specimen transformationis singularis serierum [E710], \textit{Nova
 Acta Academiae Scientarum Imperialis Petropolitinae} \textbf{12} (1801),
 58--70, reprinted in \emph{Opera Omnia, Ser.~1}, Vol.~16-2, Birkh\"auser, 1935, 41--55.

\bibitem{Gauss:1813fk}
Gauss C.F., Disquisitiones generales circa seriem inf\/initam $1 + \frac {\alpha
 \beta} {1 \cdot \gamma}x + \frac {\alpha (\alpha+1) \beta (\beta+1)} {1 \cdot
 2 \cdot \gamma (\gamma+1)} xx + \text{etc.}$, \textit{Commentationes societatis regiae scientiarum Gottingensis recentiores} \textbf{2} (1813),
1--46, reprinted in \emph{Gesammelte Werke}, Bd.~3, K.~Gesellschaft der Wissenschaften, G\"ottingen, 1876, 123--163,
 207--229.

\bibitem{iliev2006handbook}
Iliev B.Z., Handbook of normal frames and coordinates, \href{http://dx.doi.org/10.1007/978-3-7643-7619-2}{\textit{Progress in
 Mathematical Physics}}, Vol.~42, Birkh\"auser Verlag, Basel, 2006.

\bibitem{Lee:2013ys}
Lee J.M., Introduction to smooth manifolds, \textit{Graduate Texts in
 Mathematics}, Vol.~218, 2nd ed., Springer, New York, 2013.

\bibitem{Pringsheim:1911ed}
Pringsheim A., Faber G., Molk J., Analyse Alg\'ebrique, in Encyclop\'edie des
 sciences math\'ematiques, Tome~II, Vol.~II, Gauthier-Villars, Paris, 1911.

\bibitem{Rainville:1945uq}
Rainville E.D., The contiguous function relations for {${}_pF_q$} with
 applications to {B}ateman's {$J_n^{u,v}$} and {R}ice's {$H_n(\zeta,p,v)$},
 \href{http://dx.doi.org/10.1090/S0002-9904-1945-08425-0}{\textit{Bull. Amer. Math. Soc.}} \textbf{51} (1945), 714--723.

\bibitem{Szmytkowski:2007gf}
Szmytkowski R., Closed forms of the {G}reen's function and the generalized
 {G}reen's function for the {H}elmholtz operator on the {$N$}-dimensional unit
 sphere, \href{http://dx.doi.org/10.1088/1751-8113/40/5/009}{\textit{J.~Phys.~A: Math. Theor.}} \textbf{40} (2007), 995--1009.

\bibitem{Vidunas:2003kx}
Vid{\=u}nas R., Contiguous relations of hypergeometric series,
 \href{http://dx.doi.org/10.1016/S0377-0427(02)00643-X}{\textit{J.~Comput. Appl. Math.}} \textbf{153} (2003), 507--519,
 \href{http://arxiv.org/abs/math.CA/0109222}{math.CA/0109222}.

\bibitem{VorsselmanndeHeer1833}
Vorsselmann~de Heer P.O.C., Specimen Inaugurale De Fractionibus Continuis,
 Altheer, 1833, {a}vailable at \url{http://eudml.org/doc/204259}.

\bibitem{yamasuge1957}
Yamasuge H., Maximum principle for harmonic functions in {R}iemannian
 manifolds, \textit{J.~Inst. Polytech. Osaka City Univ. Ser.~A.} \textbf{8}
 (1957), 35--38.

\end{thebibliography}
\end{document}